\documentclass[review]{elsarticle}

\usepackage{lineno,hyperref}
\usepackage{graphicx}
\usepackage{amsthm}
\usepackage{amssymb}
\usepackage{amsmath}
\usepackage{bbm}
\usepackage{xcolor}
\usepackage{subcaption}

\newtheorem{prop}{Proposition}

\newtheorem{rem}{Remark}
\modulolinenumbers[5]

\journal{Journal of \LaTeX\ Templates}

\begin{document}

\begin{frontmatter}

\title{On pattern formation in reaction-diffusion systems containing self- and cross-diffusion}
\tnotetext[mytitlenote]{Fully documented templates are available in the elsarticle package on \href{http://www.ctan.org/tex-archive/macros/latex/contrib/elsarticle}{CTAN}.}

\author{Benjamin Aymard\fnref{myfootnote}}
\address{MathNeuro Team, Inria Sophia Antipolis Méditerranée, Sophia Antipolis Cedex, France}




\begin{abstract}
In this article we propose a unified framework in order to study reaction-diffusion systems containing self- and cross-diffusion using a free energy approach. 
This framework naturally leads to the formulation of an energy law, and to a numerical method respecting a discrete version of the latter.
It constitutes an alternative method and complements the standard linear stability analysis,
as it allows for the numerical study of nonlinear patterns, while monitoring the energy evolution, even in complex geometries. 
As an application, we propose and study a modified Gray-Scott system augmented with self- and cross-diffusion terms.
Numerical simulations unveil original patterns, clearly distinct from those obtained with linear diffusion only.
\end{abstract}


\end{frontmatter}


\section*{Introduction} 

\addcontentsline{toc}{section}{Introduction} 

Reaction-diffusion systems are present in numerous domains ranging from physics to chemistry to biology.
They are known to generate two 
notable 
classes of solutions: Turing patterns and travelling waves.
Since the seminal work of Alan Turing in 1952 \cite{turing},
it is known that the introduction of linear diffusion in dynamical systems may lead, under certain conditions,
to the formation of instabilities known as Turing patterns.
Turing's predictions have been observed in experiments, 
such as the Belousov-Zhabotinskii reaction 
\cite{BZ1} and \cite{BZ2}, 
the CIMA reaction \cite{CIMA},
and in many biological applications such as chemotaxis and tumor growth \cite{KS}, skin pattern formation \cite{GM} or in neuroscience \cite{patternFN}.
Travelling waves occur, for instance, in the study of systems such as combustion \cite{combustion}, 
fire propagation \cite{fire}, and ecology \cite{fisher}. Density evolution through a reactive medium can even trigger a branching process, leading to tree structure formation \cite{branching}.

Reaction-diffusion systems can be generalized by including nonlinear diffusion terms:
self-diffusion, corresponding to a diffusion proportional to the density,
and cross-diffusion, a tendency to diffuse along the gradient of another coexisting density.
Such generalization is particularly relevant in biology, for instance in predator-prey systems (notably the so-called SKT system \cite{SKT}) and in epidemiology models, where species tend to avoid crowding or avoid others.
It is also relevant in chemistry, with the tendency of certain species to repel each other.

This generalization has been widely studied in recent years \cite{mimura}, \cite{martinez}, \cite{boris}, \cite{trescases}, and pattern formation has been reported upon \cite{gambino}, \cite{turingSKT}, \cite{turingSIR}, \cite{turingSelf} (among others).
Standard linear stability analysis is a powerful tool for their analysis.
However, such analysis is only local, and even though a dispersion relation could predict instabilities and pattern formation, 
it cannot geometrically describe nonlinear patterns, nor can it describe interactions with domain borders, 
monitor the energy evolution of the full nonlinear system, or study the interaction between cross-diffusion and self-diffusion.
Moreover, there could be cases where nonlinear effects could be dominant, making the linear analysis irrelevant. 

In this work we provide a unified framework for the study of 
reaction-diffusion systems containing self- and cross-diffusion in the general case with $M$ interacting species.
To the best of our knowledge, such a framework does not exist. 
Our approach follows an energetic point of view and is inspired by diffuse-interface theory.
The most prominent example is the Cahn-Hilliard system \cite{CH} proposed in 1958:
one equation describes density evolution, with a diffusion along the gradient of a so-called chemical potential,
and one equation defines the chemical potential as a functional derivative of a free energy.
This formulation gives a natural framework for writing a free energy law, and in designing a finite-element method based on a mixed formulation \cite{JCP2013}. During the design of our numerical method, specific attention has been given to two challenging points: energy stability and accuracy.

The classical Gray-Scott (GS) system \cite{GS}, 
introduced in 1983,
is known as a good example for generating a wide variety of patterns \cite{pearson}. 
This system may be considered as a reaction-diffusion generalization 
of a special case of the Sel'kov model of glycolysis from 1968 \cite{Selkov}.
To exemplify our methodology, we introduce a generalization of the GS system, augmented with self- and cross-diffusion terms.
We study how Turing patterns and travelling waves differ when they are generated with different combinations of self- and cross-diffusion,
compared to classical linear diffusion, and how geometry 
affects
 pattern formation, while monitoring energy evolution.

The article is organized as follows.
In the first section, we propose a unified framework for studying reaction-diffusion systems containing self- and cross-diffusion terms,
we establish the corresponding energy law verified by solutions of the system, 
and we propose 
a
generalization of the GS system, augmented with self and cross-diffusion terms.
The second section is devoted to the design of the numerical method which approaches solutions to this model,
and the demonstration of a discrete version of the energy law.
The last section presents numerical simulations of the generalized GS system, with formation of original patterns in the presence of different combinations of linear, self- and cross-diffusion, within various geometries.

\section{Model}

In this section we introduce a unified framework for reaction-diffusion problems containing self- and/or cross-diffusion terms.
Interested readers may find an introduction to the topic of classical reaction-diffusion systems in \cite{perthame}.

\subsection{Unified framework}

We consider the dynamics of M interacting species in a bounded domain $\Omega \subset \mathbb{R}^d$ with $d$ the spatial dimension.
Let us denote $\phi_i(\textbf{X},t)$ the density of species $i$ at position $\textbf{X}$ and time $t$.
The evolution of the system is described by:

\begin{align}
\forall i \in \{1,...,M\}, \quad 
& \frac{\partial \phi_i}{\partial t} - \nabla^2 \mu_i = R_i,\notag \\
& R_i = R_i(\phi_1,...,\phi_M),\notag \\
& \mu_i = \left( d_{i} + d_{ii} \phi_i^{\alpha_i} + \sum_{j \not= i}d_{ij} \phi_j^{\beta_{ij}} \right) \phi_i,
\label{PDE}
\end{align}
with $R_i$ the reaction terms, and $\mu_i$ may be called chemical potentials in the framework of thermodynamics, 
if they derive from a free energy.
This system involves several levels of diffusion: 
a linear diffusion for each species, tuned by parameters $d_{i} \geq 0$, describing the tendency to fill up the space, 
a self-diffusion for each species, tuned by parameters $d_{ii} \geq 0$ and power $\alpha_i \geq 0$, describing the tendency of each species to avoid crowding, 
and cross-diffusion terms between species, tuned by parameters  $d_{ij} \geq 0$ and powers $\beta_{ij} \geq 1$, describing the tendency of each species to avoid another.
In order to mathematically close this problem, we have to add initial conditions,
and boundary conditions under the form of confinement in the domain $\Omega$:

\begin{align}
\forall i \in \{1,...,M\}, \quad
& \phi_i(\textbf{X},t=0) = \phi_i^0(\textbf{X}),\notag \\
& \nabla \mu_i \cdot n = 0 \mbox{ on }\partial \Omega,
\label{BC}
\end{align}

with $n$ the outward unit normal vector. 
If other types of boundary conditions are considered, the variational formulation and the numerical schemes might be adapted accordingly.
The reader may note that models \cite{gambino}, \cite{turingSKT}, \cite{turingSIR}, \cite{turingSelf} may be written under the form of (\ref{PDE}).
Multiplying (\ref{PDE}) by test functions in $H^1(\Omega)$,
integrating over the physical domain $\Omega$, and using integration by parts, we get the weak form of the system:  

\begin{align}
\forall i \in \{1,...,M\}, \forall \theta_i, \nu_i & \in H^1(\Omega),\notag \\
& \left( \frac{\partial \phi_i}{\partial t},\theta_i \right) + (\nabla\mu_i , \nabla \theta_i) 
- (R_i,\theta_i) = 0,\notag\\
& (\mu_i,\nu_i) = \left( \left( d_{i} + d_{ii} \phi_i^{\alpha_i} + \sum_{j\not=i}d_{ij} \phi_j^{\beta_{ij}}\right) \phi_i,\nu_i \right),
\label{weakForm}
\end{align}

where $(f,g) = \int_{\Omega}fg$ denotes the scalar product.
The mass of each species $i$ is defined as the integral of its density over the physical domain:
\begin{equation}
M_i(t) = \int_{\Omega}\phi_i(\textbf{X},t).
\label{M}
\end{equation}

\begin{prop} 
The mass variation of the solutions of (\ref{weakForm}) satisfies the mass law:
\begin{equation}
\frac{d}{dt} M_i(t) =  \int_{\Omega}R_i.
\label{Mlaw}
\end{equation}
\end{prop}

\begin{proof} 
Consider $\theta_i = 1$ in weak formulation (\ref{weakForm}).
\end{proof}

Let us define the differential form:
\begin{equation}
\omega = \sum_{j=1}^M \mu_j d \phi_j.
\label{omega}
\end{equation}
We can decompose this differential form by separating the exact part and the non-exact part:
\begin{equation}
\omega = d F + \sum_{j=1}^M g_j d \phi_j.
\label{diffForm}
\end{equation}

\begin{prop} 

If the differential form $\omega$ (\ref{omega}) is exact, solutions of system (\ref{weakForm}), if they exist, admit an energy law of the form:

\begin{equation}
\frac{d}{dt}
\left(E(\phi_1(t),...,\phi_M(t))\right)
= 
- \sum_{i=1}^M \| \nabla \mu_i \|^2_{L^2(\Omega)}
+ \sum_{i=1}^M (R_i,\mu_i),
\label{Elaw}
\end{equation}

with:

\begin{equation}
E(\phi_1(t),...,\phi_M(t)) 
=
\int_{\Omega}
F(\phi_1(t),...,\phi_M(t)) .
\label{E}
\end{equation}

\end{prop}

\begin{proof}
Let us consider $\theta_i =\mu_i$ in (\ref{weakForm}).
For all $i \in \{1,...,M\}$:
\begin{eqnarray*}
\left( \frac{\partial \phi_i}{\partial t},\mu_i \right) + (\nabla\mu_i , \nabla \mu_i) - (R_i,\mu_i) = 0.
\end{eqnarray*}

By summing all equations we get:

\[
\sum_{i=1}^M \left(\mu_i,\frac{\partial \phi_i}{\partial t}\right)
=
- \sum_{i=1}^M (\nabla \mu_i , \nabla \mu_i)
+ \sum_{i=1}^M(R_i,\mu_i).
\]

We conclude using the decomposition (\ref{diffForm}):
\[
\sum_{i=1}^M\left(\mu_i,\frac{\partial \phi_i}{\partial t}\right)
= \sum_{i=1}^M \left(\frac{\partial F}{\partial \phi_i},\frac{\partial \phi_i}{\partial t}\right)
= \frac{d}{dt}\left( F,1 \right).
\]
\end{proof}

\begin{rem} 
In the purely diffusive case ($\forall i, R_i = 0$), the free energy spontaneously decreases until reaching a state with null gradient of potential $\mu_i$.
\end{rem}

\begin{rem}
In the symmetric case $d_{ij}=d_{ji}$ and $\beta_{ij} = \beta_{ji} = 2$, the energy reads:
\begin{equation}
E(\phi_1,...,\phi_M) = 
\int_{\Omega}\
\left(\sum_{i=1}^M
d_{ii}
\frac{\phi_i^{\alpha_i+2}}{\alpha_i+2} 
+  d_i\frac{\phi_i^2}{2} 
+ \sum_{j>i} d_{ij}\frac{\phi_i^2 \phi_j^2}{2}
\right).
\label{symLaw}
\end{equation}
This case is relevant in physics: because of the principle of reciprocity, the formula is symmetric.
\end{rem}

\begin{prop} 
Solutions of the asymmetric system with $M=2$, $d_{21} = 0$ and $\beta_{12} = 1$, if they exist, satisfy the energy  law:
\begin{equation}
\frac{d}{dt} E(\phi_1,\phi_2) = 
- \sum_{i=1}^2 \|\nabla \mu_i \|^2 
+ \sum_{i=1}^2 (R_i,\mu_i) 
-\frac{d_{12}}{2} \left(\nabla \mu_2 , \nabla \phi_1^2 \right)
+\frac{d_{12}}{2} \left( R_2, \phi_1^2\right),
\label{asymElaw}
\end{equation}
with
\begin{equation}
E(\phi_1,\phi_2) = 
\int_{\Omega} 
\sum_{i=1}^2 \left( d_i \frac{\phi_i^2}{2} + d_{ii} \frac{\phi_i^{\alpha_i+2}}{\alpha_i+2} \right) 
+ \int_{\Omega}  d_{12} \phi_2 \frac{\phi_1 ^2}{2}
\label{aSymE}
\end{equation}
\end{prop}

\begin{proof}
Differentiating energy (\ref{aSymE}) with respect to time, we get:
\[
\frac{dE}{dt} = 
\left(\mu_1 ,\frac{\partial \phi_1}{dt}\right)
+ \left(\mu_2 ,\frac{\partial \phi_2}{dt}\right)
+ d_{12}\left(\frac{\phi_1^2}{2}, \frac{\partial \phi_2}{\partial t}\right)
\]

Considering $\theta_i = \mu_i, \nu_i = \frac{\partial \phi_i}{\partial t}$ in variational formulation (\ref{weakForm}), we get:

\[
\frac{dE}{dt} = 
- \sum_{i=1}^2 \|\nabla \mu_i \|^2 
+ \sum_{i=1}^2 (R_i,\mu_i) 
+ d_{12}\left(\frac{\phi_1^2}{2}, \frac{\partial \phi_2}{\partial t}\right)
\]
Finally, for the uncompensated term, considering $\theta_2 = \frac{\phi_1^2}{2}$ in (\ref{weakForm}), we get:
\[
\left(\frac{\phi_1^2}{2}, \frac{\partial \phi_2}{\partial t}\right)
=
-\frac{1}{2} \left(\nabla \mu_2 , \nabla \phi_1^2 \right)
+\frac{1}{2} \left( R_2, \phi_1^2\right)
\]
\end{proof}

\begin{rem}
This proposition may be easily generalized to cases with $M \geq 2$, as asymmetric cross-diffusion terms only affect two species $i$ and $j$ with $i \not= j$ in (\ref{weakForm}).
\end{rem}

\begin{rem}
In the purely diffusive case ($\forall i, R_i = 0$), without further assumptions, 
the sign of the uncompensated terms $\left(\nabla \mu_2 , \nabla \phi_1^2 \right)$ cannot be controlled a priori.
The proof of whether or not the free energy spontaneously decreases in this case remains open to discussion. 
\end{rem}

\begin{rem}
This asymmetric case corresponds to the so-called triangular case in SKT literature \cite{SKT}.
This case is relevant in biology, where interactions could be asymmetric.
The reader may think for instance about predator-prey systems, where prey tend to avoid predators but not the other way around, or epidemiology models, where vulnerable individuals may avoid crowds or others already infected.
\end{rem}

\subsection{Application: an augmented Gray-Scott system}

Here, we introduce an original model, derived from the classical Gray-Scott system, 
in order to show the capabilities of our proposed framework.

The Gray-Scott system \cite{GS} is a generic reaction-diffusion system known to contain rich dynamics, 
and able to generate a wide variety of patterns (see \cite{pearson}).
A typical application is the chemical reaction of three components in a stirred tank:
a reactant U, a catalyst V, and an inert product P verifying $U+V+P=1$.
From the chemistry point of view, the dynamics reads:
\begin{align*}
U + 2V \rightarrow 3V, \\
V \rightarrow P.
\end{align*}
Let us denote by $\phi_1$ (resp. $\phi_2$) the concentration of $U$ (resp. $V$). The classical system reads:
\begin{align}
\frac{\partial \phi_1}{\partial t} - \nabla^2 \phi_1 &= -\phi_1\phi_2^2 + F(1-\phi_1), \notag\\
\frac{\partial \phi_2}{\partial t} - \nabla^2 \phi_2 &= \phi_1\phi_2^2 - (F+k)\phi_2.\label{GrayScott}
\end{align}

The reaction term is composed of a replenishment term $F(1-\phi_1)$ for reactant $U$,
an exchange term $\phi_1\phi_2^2$ and a sink term $- (F+k)\phi_2$.

\begin{rem}
In order to avoid confusion, we call the term $\phi_1\phi_2^2$ an exchange term.
However, the reader may note that in the Gray-Scott literature, this term is usually referred to as a reaction term.
\end{rem}

We propose to consider a modification of this system, by adding self- and cross-diffusion terms:
\begin{align}
\frac{\partial \phi_1}{\partial t} - \nabla^2 \mu_1 &= -\phi_1\phi_2^2 + F(1-\phi_1), \notag \\
\frac{\partial \phi_2}{\partial t} - \nabla^2 \mu_2 &= \phi_1\phi_2^2 - (F+k)\phi_2,\notag\\
\mu_1 = (d_{1} &+ d_{11} \phi_1^{\alpha_1} + d_{12} \phi_2^{\beta_{12}})\phi_1,\notag\\
\mu_2 = (d_{2} &+ d_{22} \phi_2^{\alpha_2} + d_{21} \phi_1^{\beta_{21}})\phi_2. \label{mGrayScott}
\end{align}
Note that this system has the form of (\ref{PDE}). 
We will numerically investigate the dynamics of this system in the last section using the methodology developed in this article.

\section{Numerical Method}

In this section we derive a numerical method in order to simulate system (\ref{PDE}),
following the methodology developed in \cite{JCP2013} and \cite{JCP2019}.
Interested readers may find an introduction to numerical methods in \cite{Allaire}.
Numerical simulations were done using the free open source software FreeFem++ \cite{freefem} version 4.6.

\subsection{Operator for nonlinear diffusion}

In this paragraph we introduce a numerical method in order to simulate the nonlinear diffusion part of the system. 
Let us denote by $\Delta t^n$ the time step at iteration $n$, $T>0$ a final time, 
and $\phi_i^n$ the density of species $i$ at time $t^n = \sum_{i=0}^n \Delta t^i$.
Knowing that
in the symmetric case (where the differential form $\omega$ defined by (\ref{omega}) is exact), 
at continuous level, the system respects energy law (\ref{Elaw}), we design our method such that a similar property is verified at the discrete level. 
Starting with a Taylor-Lagrange expansion of the energy 

and using the definition of $E$ and $\mu_i$, we get, using multi-index notation:

\begin{align*}
E(\phi + \delta)
= E(\phi) 
+ \int_{\Omega} \sum_{i=1}^M \mu_i \delta_i
+  \int_{\Omega} \frac{1}{2} \sum_{i=1}^M \sum_{j=1}^M \frac{\partial \mu_i}{\partial \phi_j}\delta_i \delta_j
+ \int_{\Omega} \frac{1}{6} \sum_{i=1}^M \sum_{j=1}^M \sum_{k=1}^M R_{i,j,k} \delta_i \delta_j \delta_k,
\end{align*}

with $\delta = (\delta_1,...,\delta_M)$ and

\[
|R_{i,j,k}| \leq \max_{i,j,k} \left\| \frac{\partial^2 \mu_i}{\partial \phi_j \partial \phi_k} \right\|_{C([0,T],L^{\infty}(\Omega))}.
\] 

Therefore, energy variation reads:

\begin{align}
\frac{E^{n+1} - E^n}{\Delta t}
= 
& \sum_{i=1}^M \int_{\Omega} \left(\frac{\phi_i^{n+1} - \phi_i^n}{\Delta t}\right)
\left(
\mu_i^n + \frac{1}{2}\sum_{j=1}^M (\phi^{n+1}_j - \phi^n_j)\frac{\partial \mu_i^n}{\partial \phi_j}
\right)\\
+ 
& 
\frac{\Delta t^2}{6} \sum_{i=1}^M \sum_{j=1}^M \sum_{k=1}^M 
\int_{\Omega}  R_{i,j,k} 
\frac{(\phi_i^{n+1} - \phi_i^n)}{\Delta t}
\frac{(\phi_j^{n+1} - \phi_j^n)}{\Delta t}
\frac{(\phi_k^{n+1} - \phi_k^n)}{\Delta t}.
\label{Evariation}
\end{align}

We define the operator $B_{\Delta t}$ approximating the diffusive part as the solution of the variational problem:
\begin{align}
\forall i \in & \{1,...,M\},\forall \theta_i, \nu_i \in H^1(\Omega), \notag\\
& \left( \displaystyle\frac{B_{\Delta t}(\phi_i^{n}) - \phi_i^n}{\Delta t},\theta_i \right) + (\nabla\mu_i^{n+1/2} , \nabla \theta_i) = 0,\notag\\
& (\mu_i^{n+1/2},\nu_i) = \left( \mu_i^n + \frac{1}{2}\sum_{j=1}^M (\phi^{n+1}_j - \phi^n_j)\frac{\partial \mu_i^n}{\partial \phi_j}, \nu_i \right).
\label{operatorB}
\end{align}

Defining a mesh $T_h$ of the domain $\Omega$ and replacing $H^1(\Omega)$ by discrete finite element achieves the space discretization.
The variational formulation is reduced to a linear system that we solve at each time step.
In practice, we use an adaptive mesh refinement strategy.
We refer the reader to \cite{freefem} for the details in Freefem++ mesh adaptation strategy.
Using the function \textit{adaptmesh}, we refine the mesh locally according to variation of every component of both $\phi$ and $\mu$, with a maximum triangle size of $h_{\max}$ and a minimum set to $h_{\min}$ (parameter values will be specified in the next section).

\begin{rem}
Conservation of positivity for operator $B_{\Delta t}$ has been observed numerically in all the simulations done for this article. However, the proof remains open.
\end{rem}

\begin{prop} 
In the symmetric case where the differential form $\omega$ (\ref{omega}) is exact,
the numerical operator $B_{\Delta t}$, defined by (\ref{operatorB}), respects a discrete energy law of the form:
\begin{equation}
\frac{E^{n+1}-E^{n}}{\Delta t} 
= - \sum_{i=1}^M \| \nabla \mu_i^{n+1/2}\|^2  + \epsilon,
\label{errorB}
\end{equation}
with:
\[
\epsilon = O (\Delta t^2).
\]
\end{prop}

\begin{proof} 
Let us denote:
\[
\tilde{F}_i = 
\mu_i^n + \frac{1}{2}\sum_{j=1}^M (\phi^{n+1}_j - \phi^n_j)\frac{\partial \mu_i^n}{\partial \phi_j}.
\]
We consider $\theta_i =\mu_i^{n+1/2},  \nu_i = \frac{\phi_i^{n+1} - \phi_i^n}{\Delta t}$ in the numerical scheme (\ref{operatorB}).
Summing up all the equations and cancelling cross terms we obtain:
\[
\sum_{i=1}^M \left(\tilde{F_i},\frac{\phi_i^{n+1} - \phi_i^n}{\Delta t} \right) = - \sum_i \| \nabla \mu_i^{n+1/2}\|^2.
\]
Therefore $\epsilon$ reads:
\begin{align*}
\epsilon & = \frac{E^{n+1} - E^n}{\Delta t} + \sum_i \| \nabla \mu_i^{n+1/2}\|^2\\
& = \frac{E^{n+1} - E^n}{\Delta t} - \sum_{i=1}^M \left(\tilde{F_i},\frac{\phi_i^{n+1} - \phi_i^n}{\Delta t} \right).
\end{align*}
Finally, comparing with (\ref{Evariation}), we see that:
\[
|\epsilon| \leq 
\Delta t^2 \frac{M^3}{6}
\max_{i,j,k}
\left\| \frac{\partial^2 \mu_i}{\partial \phi_j \partial \phi_k} \right\|_{C([0,T],L^{\infty}(\Omega))}
\max_i
\left\| \frac{\partial \phi_i}{\partial t} \right\|_{C([0,T],L^{3}(\Omega))}.
\]
\end{proof}

\begin{rem}
Operator $B_{\Delta t}$, constructed for the symmetric case, can be extended to the asymmetric case, as it is clearly consistent with system (\ref{weakForm}).
However, the proof of second order accuracy in time, similar to the symmetric case, remains open.
\end{rem}

\subsection{General scheme}

We propose a numerical method based on a splitting strategy, solving sequentially the reaction part of the equation, then the diffusion part.
Using a Strang splitting \cite{strang} ensures a second order accuracy in time, as long as the two components are second order accurate with respect to time.

For the reaction part, we use a classical second order method, namely the Heun method.
Let us denote by $A_{\Delta t}$ the operator defined by:

\begin{align}
\forall i \in \{1,...,M\}, \quad
& \phi_i^{*} = \phi_i^n + \Delta t  R_i(\phi_1^n,...,\phi_M^n),\notag \\
& A_{\Delta t}(\phi_i^{n}) = \phi_i^n + \frac{\Delta t}{2} \left( R_i(\phi_1^n,...,\phi_M^n) + R_i(\phi_1^*,...,\phi_M^*) \right).
\label{operatorA}
\end{align}

Using operators $A_{\Delta t}$ and $B_{\Delta t}$ the Strang splitting reads:

\begin{equation}
\phi^{n+1} = A_{\Delta t/2} B_{\Delta t} A_{\Delta t/2} (\phi^n).
\label{scheme}
\end{equation}

Let us define the useful notations: 
$R_i^1 = R_i(\phi_i^n)$, $R_i^2 = R_i(\phi_i^n + \Delta t R_1)$,
$R_i^3 = R_i(B_{\Delta t} (A_{\Delta t/2} (\phi^n)))$ and $R_i^4 = R_i(B_{\Delta t} (A_{\Delta t/2} (\phi^n) + \Delta t R_3))$.

\begin{prop} 
The numerical scheme satisfies a discrete mass law similar to (\ref{Mlaw}) under the form:
\[
\forall n \geq 1, \int_{\Omega} \phi_i^{n+1} 
= \int_{\Omega} \phi_i^n 
+ \frac{\Delta t}{4}(R_i^1 + R_i^2 + R_i^3 + R_i^4).
\]
In particular, in the purely diffusive case $\forall i \in \{1,...M\}, R_i=0$, the mass is conserved.
\end{prop}

\begin{proof}
Consider $\theta_i = 1$ in (\ref{operatorB}) and combine with (\ref{operatorA}).
\end{proof}

\begin{prop} 
When the differential form $\omega$ (\ref{omega}) is exact, the numerical scheme (\ref{scheme}) respects a discrete energy law analogous to (\ref{Elaw}) of the form:
\begin{equation}
\frac{E^{n+1} - E^n}{\Delta t} 
= 
- \sum_{i=1}^M \| \nabla \mu_i^{n+1/2}\|^2  
+ \left(\frac{R_i^1 + R_i^2 + R_i^3 + R_i^4}{4},\mu_i^{n+1/2}\right)
+ \xi
,
\label{discreteElaw}
\end{equation}
with:
$\xi = \frac{E^{n+1} - E^n}{\Delta t} - 
\sum_{i=1}^M \left( \tilde{F}_i \left(A_{\Delta t}(\phi_i^n),B_{\Delta t} (A_{\Delta t}(\phi_i^n))\right),
\frac{\phi_i^{n+1} - \phi_i^n}{\Delta t} \right)$.
Error term $\xi$ tends to $0$ as $\Delta t$ tends to $0$.
\end{prop}

\begin{proof}
Let us write numerical scheme (\ref{scheme}) in an alternative form:
\begin{align*}
\forall i \in \{1,...,M\}, \\
(\phi_i^{n+1},\theta_i) 
 &= (\phi_i^n,\theta_i) + 
\Delta t \left(\frac{R_i^1 + R_i^2 + R_i^3 + R_i^4}{4},\theta_i\right) 
- \Delta t(\nabla \mu_i^{n+1/2},\nabla \theta_i),\\
(\mu_i^{n+1/2},\nu_i) &= \left( \tilde{F}_i \left(A_{\Delta t}(\phi_i^n),B_{\Delta t} (A_{\Delta t}(\phi_i^n))\right),\nu_i\right).\\
\end{align*}
Then, considering $\theta_i = \mu_i^{n+1/2}$, $\nu_i = \frac{\phi_i^{n+1} - \phi_i^n}{\Delta t}$ and summing all the equations we get:
\[
\sum_{i=1}^M\left(\tilde{F}_i,\frac{\phi_i^{n+1} - \phi_i^n}{\Delta t}\right)
=
\sum_{i=1}^M \left(\frac{R_i^1 + R_i^2 + R_i^3 + R_i^4}{4},\mu_i\right)
- \sum_{i=1}^M \|\nabla \mu_i^{n+1/2} \|^2.
\]
Using the definition of $\xi$, we obtain formula (\ref{discreteElaw}).
Let us remark that, for all $i \in \{1,...,M\}$ we get:
\begin{align*}
\lim\limits_{\Delta t \rightarrow 0} &
\sum_{i=1}^M\left( \tilde{F}_i \left(A_{\Delta t}(\phi_i^n),B_{\Delta t} (A_{\Delta t}(\phi_i^n))\right),\frac{\phi_i^{n+1} - \phi_i^n}{\Delta t}\right)\\
& = \sum_{i=1}^M\left(\mu_i(t^n),\frac{\partial \phi_i}{\partial t}(t^n)\right)
= \frac{dE}{dt}(t^n).
\end{align*}
Therefore we can conclude, using the definition of $\xi$, that: $\lim \limits_{\Delta t \rightarrow 0} \xi = 0$.
\end{proof}

\begin{rem}
Formula (\ref{discreteElaw}) is a coarse approximation. We only derive it in order to give the general form of the error term.
A finer approximation seems possible starting for instance from the Baker-Campbell-Hausdorff formula, 
combining the Heun error and B operator error terms (\ref{errorB}).
\end{rem}

\section{Pattern formation}

In this section we numerically investigate the augmented Gray-Scott system introduced in first section (\ref{mGrayScott}) to exemplify the methodology developed in this article.

\subsection{Effect of nonlinear diffusion terms}

In this paragraph, we study the effect of nonlinear diffusion terms on the formation of patterns in a square geometry $\Omega = [0,1]^2$.
Following \cite{pearson}, as initial conditions we consider a square, denoted by S, centered at $(0.5,0.5)$ with size $0.05$, and we define:
\begin{align}
\phi_1(t=0) &= (0.5 + 0.05 U)\mathbbm{1}_{S} + \mathbbm{1}_{\Omega \setminus S},\notag\\
\phi_2(t=0) &= (0.25 + 0.05 V)\mathbbm{1}_{S},\label{ICGS}
\end{align}
with $U,V$ random numbers uniformly distributed on $[0,1]$.
Parameters are set to: 
$\Delta t = 1$, $h_{\min} = 0.002$, $h_{\max} = 0.07$, $F=0.037, k=0.06$, $\alpha_1 = \alpha_2 = 2$.
By default, parameters that are not mentioned are set to $0$.
For tests involving self-diffusion, we have taken 
$\Delta t = 0.1$
to ensure numerical stability.

Results are displayed in Figure \ref{figRD} for the symmetric cases where $d_{12} = d_{21}$ and $\beta_{12} = \beta_{21} = 2$.
The first case ($d_1 = 2e^{-5}, d_2 = 1e^{-5}$) is a test against a known solution, and shows the classical Gray-Scott pattern for linear diffusion. We will refer to this case as the reference case. 
By slightly perturbing this reference case with self and/or symmetric cross-diffusion terms, 
we observe a qualitative change in the patterns created by the system:
cases one ($d_{11} = d_{22} = 1e^{-5}$) and two ($d_2 = 1e^{-5}, d_{11}=1e^{-5}$) generate a ring-shaped travelling wave, 
which differ by their contours, whereas cases three ($d_{1} = 2e^{-5}, d_{22} = 1e^{-5}, d_{12} = d_{21} = 1e^{-6}$) and four ($d_{1} = 3e^{-5}, d_{2} = 1e^{-5}, d_{12} = d_{21} = 5e^{-6}$) generate complex Turing-like patterns.

In Figure \ref{figRD2}, we display the results in the asymmetric cases $M=2$, $d_{21} = 0$ and $\beta_{12} = 1$.
Similarly to the symmetric cases, we slightly perturb the reference case with self- and/or asymmetric cross-diffusion terms, 
and similarly, we also observe qualitative change in pattern formation: the fifth case ($d_{11} = d_{22} = 1e^{-5}, d_{12}=1e^{-5}$), sixth case ($ d_1 = 2e^{-5}, d_{22} = 1e^{-5}$) and the seventh case ($ d_1 = 3e^{-5}, d_2 = 1e^{-5}, d_{12} = 1e^{-5}$) also generate complex Turing-like patterns. 

For each case, the energy is monitored (Figure \ref{figRDenergy}), computed using formula (\ref{symLaw}) for symmetric cases (reference case, and cases 1 to 4), 
and using formulas (\ref{aSymE}) for asymmetric cases (5 to 7). 
Interestingly, we observe that in all cases energy decreases with respect to time.
However, this could be due to a specific choice of parameter.

Classically, in system (\ref{GrayScott}), different patterns such as spots, stripes and travelling waves are obtained by varying the reactive term. By contrast, model (\ref{mGrayScott}) generates different kinds of patterns for the same reactive term.
Moreover, model (\ref{mGrayScott}) creates original patterns due to nonlinearities, not observed with linear diffusion only.

\begin{figure}[ht]
\begin{subfigure}{.5\textwidth}
  \centering
  \includegraphics[width=\linewidth]{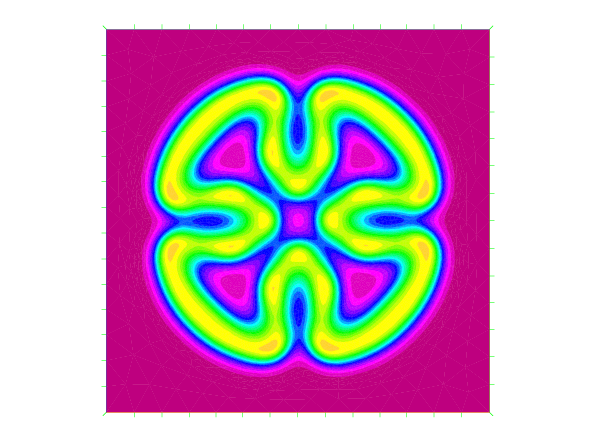}  
  \caption{Reference - classical Gray-Scott}
  \label{fig:sub-first}
\end{subfigure}
\\
\begin{subfigure}{.5\textwidth}
  \centering
  \includegraphics[width=\linewidth]{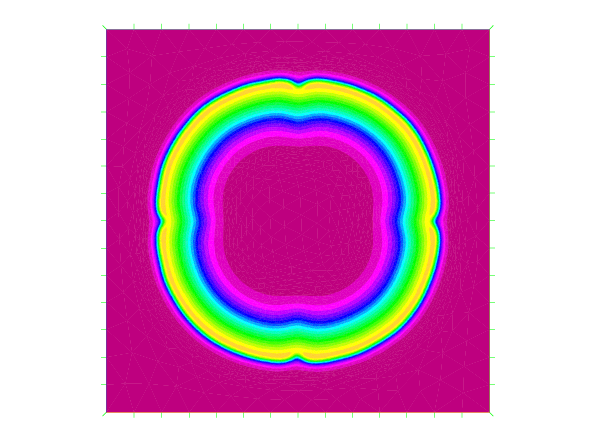}  
  \caption{Case 1}
  \label{fig:sub-first}
\end{subfigure}
\begin{subfigure}{.5\textwidth}
  \centering
  \includegraphics[width=\linewidth]{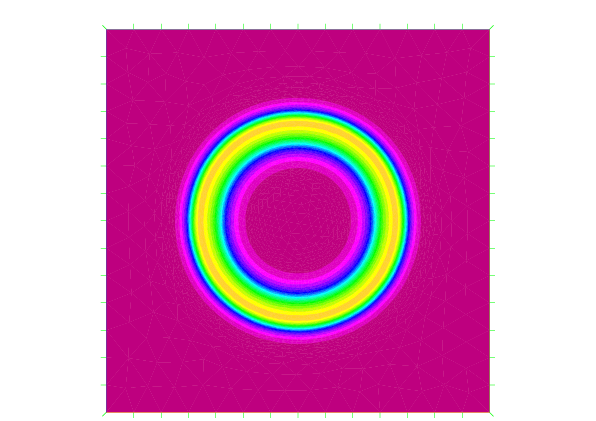}  
  \caption{Case 2}
  \label{fig:sub-second}
\end{subfigure}
\begin{subfigure}{.5\textwidth}
  \centering
  \includegraphics[width=\linewidth]{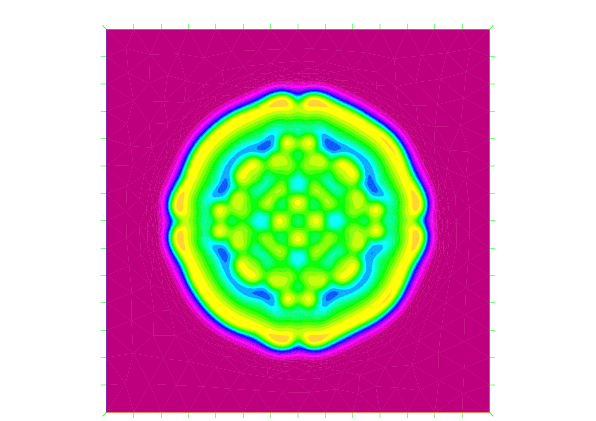}  
  \caption{Case 3}
  \label{fig:sub-second}
\end{subfigure}
\begin{subfigure}{.5\textwidth}
  \centering
  \includegraphics[width=\linewidth]{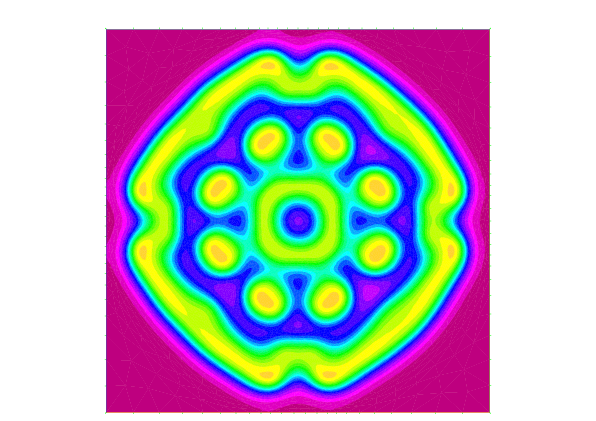}  
  \caption{Case 4}
  \label{fig:sub-second}
\end{subfigure}
\caption{Nonlinear pattern formation driven by different combinations of linear, self and symmetric cross-diffusion.
Simulations of modified Gray-Scott reaction-diffusion system (\ref{mGrayScott}) 
with boundary conditions (\ref{BC}) starting from initial conditions (\ref{ICGS}).
Plots of density $\phi_1$ (color code: continuous scale from yellow for $\phi_1 \in [0,0.25]$, to green for $\phi_1 \in [0.25,0.5]$, to blue for $\phi_1 \in [0.5,0.75]$, to purple for $\phi_1 \in [0.75,1]$).
The same reaction term triggers different patterns depending on the combinations of nonlinear diffusion coefficients.}
\label{figRD}
\end{figure}

\begin{figure}[ht]
\begin{subfigure}{.5\textwidth}
  \centering
  \includegraphics[width=\linewidth]{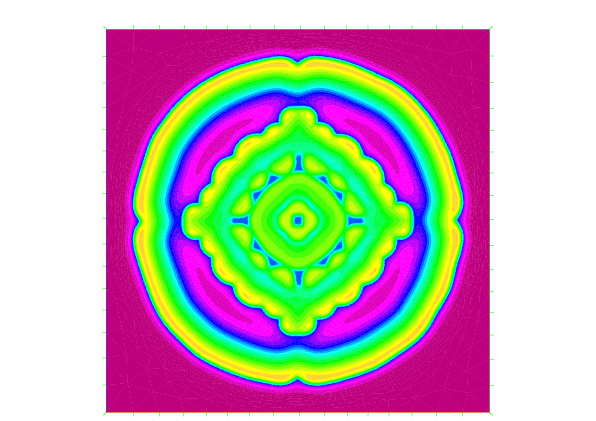}  
  \caption{Case 5}
  \label{fig:sub-first}
\end{subfigure}
\\
\begin{subfigure}{.5\textwidth}
  \centering
  \includegraphics[width=\linewidth]{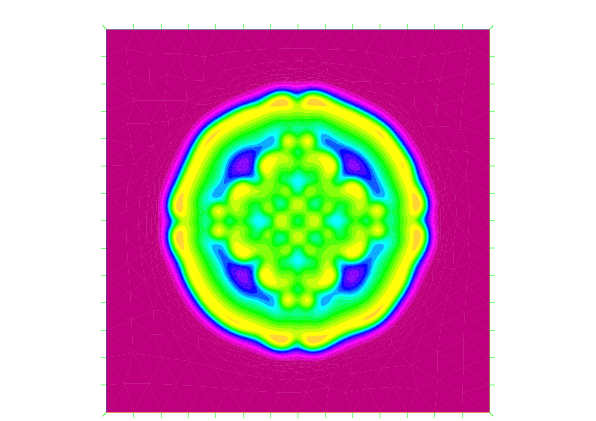}  
  \caption{Case 6}
  \label{fig:sub-first}
\end{subfigure}
\begin{subfigure}{.5\textwidth}
  \centering
  \includegraphics[width=\linewidth]{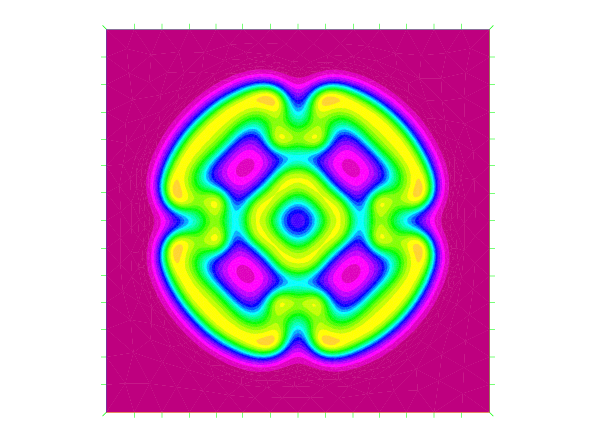}  
  \caption{Case 7}
  \label{fig:sub-second}
\end{subfigure}
\caption{Nonlinear pattern formation driven by different combinations of linear, self and asymmetric cross-diffusion.
Simulations of modified Gray-Scott reaction-diffusion system (\ref{mGrayScott}) 
with boundary conditions (\ref{BC}) starting from initial conditions (\ref{ICGS}).
Plots of density $\phi_1$ (same color code as for Figure \ref{figRD}).
Like the symmetric case, the same reaction term triggers different patterns depending on the combinations of nonlinear diffusion coefficients.}
\label{figRD2}
\end{figure}

\subsection{Effect of geometry}

In this numerical experiment, the aim is to study the effect of geometry on pattern formation.
Results are displayed on Figure \ref{figGeom}.

We consider an astroid domain with concave borders and four singularities,
using parameters of 
case 
3, then a trefoil domain with cusps, using parameters of 
case 2 
, and a rectangular domain (anisotrope), using parameters of 
case  
7.
The density is forced to adapt into a different shape, and this effect backpropagates all along the pattern until the center.
We note that the effect is local, in the sense that the evolution of the density is the same as long as it is far from the border. Eventually, the patterns differs greatly from the one observed in previous sections.

Then, we consider a not simply connex domain with a trefoil perturbation (with singularities on the borders) at the center.
Two kinds of nonlinear diffusion are tested: combination of self- and cross-diffusion, like in  
case 
5, and combination of linear diffusion and asymmetric cross-diffusion, as in 
case 
6. The pattern formation presents similarities on the front of propagation (similar to the unperturbed case), 
but clearly differs around the center, due to the perturbation.

Once again the energy is monitored for each test, computed using formula (\ref{symLaw}) for symmetric cases, 
and using formulas (\ref{aSymE}) for asymmetric cases.
Interestingly, in the case of the astroid domain, a steady state is rapidly reached, 
minimizing (\ref{symLaw}) while both densities $\phi_1$ and $\phi_2$ stop evolving with respect to time.

\begin{figure}[ht]
\begin{subfigure}{.5\textwidth}
  \centering
  \includegraphics[width=\linewidth]{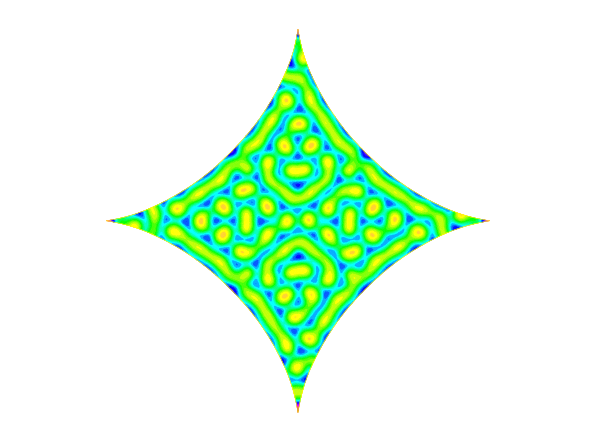}  
  \caption{Case 3 within an astroid domain}
  \label{fig:sub-first}
\end{subfigure}
\begin{subfigure}{.5\textwidth}
  \centering
  \includegraphics[width=\linewidth]{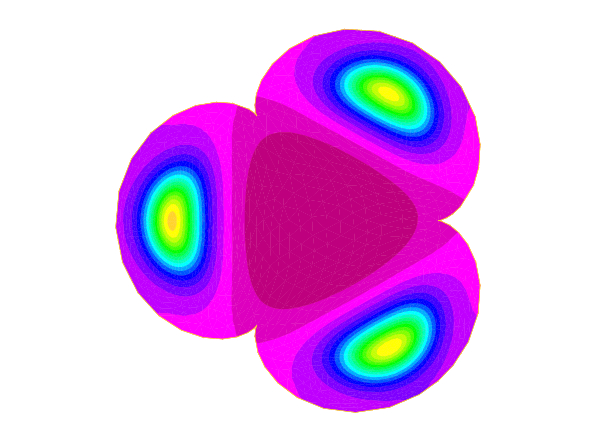}  
  \caption{Case 2 within a trefoil domains}
  \label{fig:sub-first}
\end{subfigure}
\begin{subfigure}{1.0\textwidth}
  \centering
  \includegraphics[width=\linewidth]{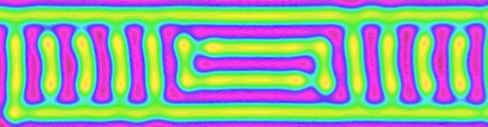}  
  \caption{Case 7 within a rectangle domain}
  \label{fig:sub-second}
\end{subfigure}
\\
\begin{subfigure}{.5\textwidth}
  \centering
  \includegraphics[width=\linewidth]{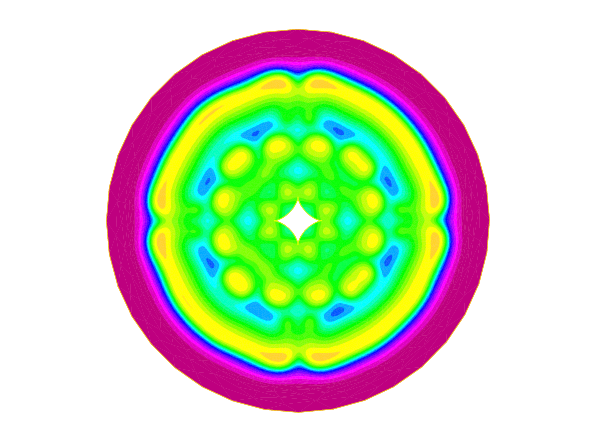}  
  \caption{Case 3 perturbed with an astroid hole}
  \label{fig: sub-third}
\end{subfigure}
\begin{subfigure}{.5\textwidth}
  \centering
  \includegraphics[width=\linewidth]{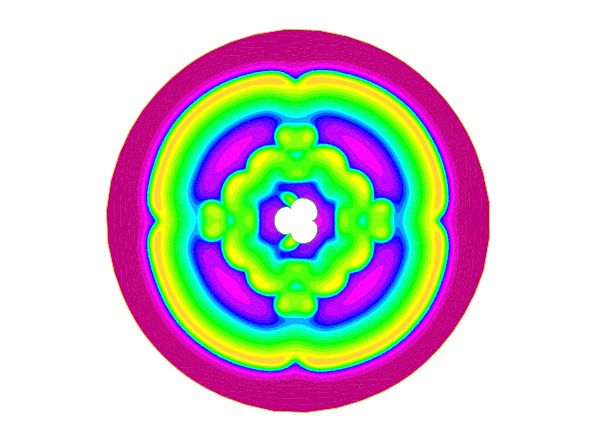}  
  \caption{Case 5 perturbed with a trefoil hole}
  \label{fig:sub-second}
\end{subfigure}
\caption{Nonlinear pattern formation within complex geometries.
Simulations of modified Gray-Scott reaction-diffusion system (\ref{mGrayScott}) 
with boundary conditions (\ref{BC}) starting from initial conditions (\ref{ICGS}).
Plots of density $\phi_1$ (same color code as for Figure \ref{figRD}).}
\label{figGeom}
\end{figure}

\begin{figure}[ht]
\begin{subfigure}{\textwidth}
  \centering
  \includegraphics[width=\linewidth]{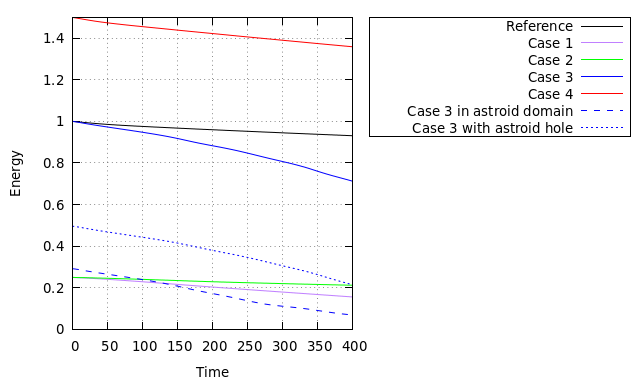}
  \caption{Symmetric cases corresponding to Figures \ref{figRD} and \ref{figGeom}. Energy computed using formula (\ref{symLaw})}
  \label{fig:sub-first}
\end{subfigure}
\begin{subfigure}{\textwidth}
  \centering
  \includegraphics[width=\linewidth]{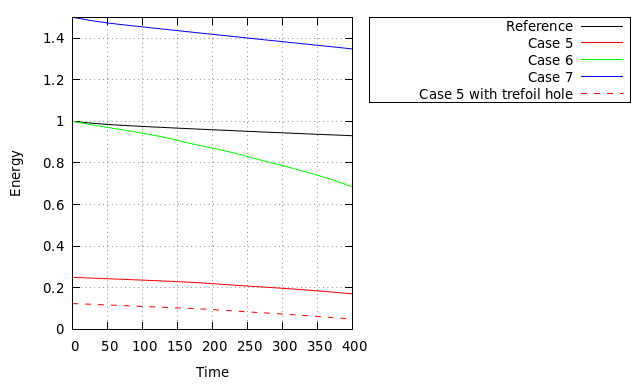} 
  \caption{Asymmetric cases corresponding to Figures \ref{figRD2} and \ref{figGeom}. Energy computed using formula (\ref{aSymE}}
  \label{fig:sub-first}
\end{subfigure}
\caption{Time evolution of energy with respect to time.
In both cases, energy is scaled with the initial energy of the reference case.
We observe that, for all tests, energy is a decreasing function of time.
However, it could be due to a specific choice of parameters.}
\label{figRDenergy}
\end{figure}

\section*{Conclusion} 

\addcontentsline{toc}{section}{Conclusion} 

In this article we have proposed a unified framework (\ref{PDE}) to study reaction-diffusion systems containing self- and cross-diffusion.
We have theoretically established an energy law 
((\ref{Elaw}),(\ref{symLaw}) for the symmetric case, and (\ref{asymElaw}),(\ref{aSymE}) for the asymmetric case) 
that was later exemplified by numerical simulations (Figure \ref{figRDenergy}).
We have derived a numerical method (\ref{scheme}) for this class of systems, 
and have shown its key property, a numerical version of the energy law (\ref{errorB}), (\ref{discreteElaw}).
Our approach complements the classical standard linear stability analysis, 
allowing the geometrical description of nonlinear patterns, 
the numerical study of boundary effects, the interaction of different kinds of nonlinear diffusion, while monitoring the energy of the system.

In order to exemplify our methodology, we have introduced a generalization of the classical Gray-Scott system, 
augmented with self- and cross-diffusion terms (\ref{mGrayScott}). 
During numerical simulations of this model, we have observed patterns that are clearly distinct from those obtained with linear diffusion, even for the same reaction term 
(Figures \ref{figRD}, \ref{figRD2} and \ref{figGeom}).
Those patterns are due to nonlinearity, and are original to our work.

As further work, a full and systematic numerical exploration, in the spirit of Pearson \cite{pearson}, 
might be done on system (\ref{mGrayScott}), in order to map the parameter space with different pattern types.
Several numerical investigations could be done such as
evaluating the effect of different kinds of boundary conditions (more general than (\ref{BC})), 
study cases with more than two species or study branching process \cite{branching} perturbed by self- and cross-diffusion.
Another interesting aspect is the interaction between nonlinear terms. 
This phenomenon could be studied in more detail using numerical tools, 
and hopefully provide insight for a more theoretical work.

\section*{Acknowledgement}
The author would like to thank Romain Veltz, Mathieu Desroches and Audric Drogoul for fruitful discussions and feedback about this work.

\bibliography{RDarxiv}

\end{document}